\documentclass{amsart}
\usepackage{amsmath}
\usepackage{amssymb}
\usepackage{amsthm}

\newtheorem{theorem}{Theorem}
\newtheorem{definition}{Definition}
\newtheorem{lemma}{Lemma}
\newtheorem{proposition}{Proposition}
\newtheorem{corollary}{Corollary}
\newtheorem{example}{Example}
\newtheorem{remark}{Remark}

\title[Sections with constant surface gravity on null hypersurfaces]{On the existence of sections with constant surface gravity on null hypersurfaces}



\begin{document}

	\author{Ivan P. Costa e Silva}
	\address{
		Department of Mathematics, Universidade Federal de Santa Catarina, 88.040-900\\
		Florianópolis-SC, Brazil}
	\email{pontual.ivan@gmail.com}
	
	\author{José L. Flores}
	\address{Departamento de Álgebra, Geometría y Topología, Facultad de Ciencias, Universidad de Málaga, Campus Teatinos, 29071  \\ Málaga, Spain}
	\email{floresj@uma.es}
	
	\author{Benjamín Olea}
	\address{	Departamento de Matemática Aplicada, Escuela de Ingenierías Industriales, Universidad de Málaga, Extensión Campus Teatinos, 29071 \\ Málaga, Spain}
	\email{benji@uma.es}

	\begin{abstract}
		We identify, in spacetimes satisfying the null convergence condition, a certain natural class of null hypersurfaces that admit null sections with constant surface gravity. Our work is meant to offer complementary results to previous work on null hypersurfaces of zero expansion arising from the study of Cauchy and black hole horizons in general relativity. 
	\end{abstract}

	\maketitle

	\newcommand{\s}{\mathcal{S}}
	\newcommand{\gs}{\kappa}  % Surface gravity
	\newcommand{\fpf}{b}  % Final parameter function
	\newcommand{\alf}{\Lambda} % affine lenght function

	\section{Introduction}\label{sec0}
	Let $(M,g)$ be a {\it spacetime}, i.e. a connected time-oriented Lorentzian manifold, with dimension $n\geq3$. A {\it null hypersurface} in $(M,g)$ is a codimension-one submanifold in $M$ with an everywhere-degenerate induced metric tensor. Given a null hypersurface $L\subset M$, there always exists a tangent null vector field unique up to rescaling by an everywhere-nonzero function on $L$. We shall refer to any such $\xi\in\mathfrak{X}(L)$ as a {\it null section} of $L$. Associated with any null section $\xi$ there exists a  function $\gs_\xi\in C^\infty(L)$ such that 
	$$\nabla_\xi\xi=\gs_\xi \cdot \xi,$$
	called the {\it surface gravity}\footnote{The reader should be warned, however, that physicists use this term in many, sometimes inequivalent, ways. See, e.g., \cite{Visser} and references therein for a thorough discussion.} of $L$ with respect to $\xi$. In particular, the above equation implies that the maximal integral curves of $\xi$ are  null pregeodesics, which are called {\it null (geodesic) generators} of $L$ when they are parametrized as geodesics.
	
	In this context, a classic, much-studied question is whether we can find a null section with {\it constant} associated surface gravity. Some known situations where an affirmative answer can be given are the following.
	\begin{itemize}
		\item There exists a conformal vector field in $(M,g)$ transverse to the null hypersurface $L$ \cite{GutOleainduced}.
		\item The null hypersurface $L$ admits a codimension two (in $M$) spacelike submanifold which causally separates it \cite[Theorem 18]{Kupeli}.
		
		\item The dominant energy condition holds in $(M,g)$
		and the null hypersurface $L$ is a Killing horizon (see \cite{GutOleaAxioms} for a short proof). This is known in the physics literature as the {\it zeroth law of black hole thermodynamics} and it is a fundamental tenet of black hole physics, providing much of the motivation behind the study of null sections with constant surface gravity. (Recall that the {\it dominant energy condition} 
		holds in a spacetime $(M,g)$ when for any future-directed timelike vector $u \in TM$ the vector metrically equivalent to $-G(u,\cdot)$ is either zero or future-directed causal, where $G$ denotes the Einstein tensor of $(M,g)$. A null hypersurface $L$ is a  {\it Killing horizon} if there is a Killing vector field $K\in\mathfrak{X}(M)$ whose restriction to $L$ is a null section thereof. In particular, this implies that $L$ is totally geodesic.)

		\item The dominant energy condition holds in $(M,g)$ and the null hypersurface is totally geodesic and compact. This result is also physically motivated, this time around by the so-called {\it cosmic censorship conjecture} in general relativity \cite{Minguzzi,Reiris}.
	\end{itemize}
	Note that in each of these situations one either assumes the spacetime $(M,g)$ has nontrivial symmetries or else that somewhat stringent topological conditions hold on the hypersurface $L$. In this work we prove the following theorem, which guarantees the existence of a null section with nonzero constant surface gravity under quite mild geometric restrictions on the spacetime and on the null hypersurface. 
	
	\begin{theorem}
		\label{teoprin}
		Let $(M,g)$ be a spacetime satisfying the {\it null convergence condition}, i.e., $Ric(u,u) \geq 0$ for any null $u\in TM$, and let $L\subset M$ be a future weakly inextendible\footnote{If $L$ is not future weakly inextendible, the conclusion of the theorem still holds if conditions (1) and (2) are satisfied by a {\em maximal} future weak extension.} null hypersurface that 
		\begin{enumerate}
			\item has future prolonged null geodesic generators, and
			\item is of future finite type.
		\end{enumerate}
		Then $L$ admits a null section with nonzero constant surface gravity.
	\end{theorem}
	
	These geometric conditions on $L$ are spelled out more precisely in section \ref{keykey}, but we briefly describe them here in informal terms. A null geodesic generator of $L$ is {\it future prolonged} when it is either future complete or else it can be extended to the future at least a little beyond $L$, maybe through an isometric extension of $(M,g)$ if necessary. By a {\it future weak extension} of the null hypersurface we mean an extension to the future as an immersed null hypersurface in $(M,g)$ ``along the null geodesic generators'', and thus $L$ is {\it future weakly inextendible} if such a future weak extension does not exist. Finally, the {\it finite type} is a demand of a mildly controlled behaviour of the eigenvalues of the null shape operator, together with
	the existence of some point with  
	strictly negative expansion (equivalently, strictly positive null mean curvature) on each null generator of $L$. While this last condition excludes totally geodesic null hypersurfaces from the scope of our result, a simple but very important non-totally geodesic example where all the hypotheses of Theorem \ref{teoprin} are fulfilled is a past-lightcone in Minkowski spacetime.

	Observe that if we rescale the null section $\xi \in \mathfrak{X}(L)$ by nowhere-zero function $f\in C^\infty(L)$, then the surface gravity changes in a simple way. In fact, the new surface gravity is just
	$\gs_\xi f+df(\xi)$. Thus, the problem of finding a null section with constant surface gravity is equivalent to finding a nowhere-zero global solution $f\in C^{\infty}(L)$
	to the first-order linear PDE 
	\begin{equation}\label{basicpde}
		\xi(f) + \gs_\xi\,f=c
	\end{equation}
	for some constant $c\in \mathbb{R}$. This can, of course, always be done locally (for any given constant $c$), but it is a globally nontrivial issue. Indeed, a global solution may not exist: in Example \ref{ej2}  a totally geodesic null hypersurface is exhibited which does not admit any globally defined null section with constant (zero or nonzero) surface gravity.
	
	Then, there is the issue of which constants $c\in \mathbb{R}$ can arise as surface gravities on $L$. By further rescaling it is easy to see that only the cases $c\neq0$ and $c=0$ need to be considered. Example \ref{ejemplocono} shows that it may be perfectly possible to find a null section of zero constant surface gravity {\it and} another with nonzero constant surface gravity on the same null hypersurface, while Examples \ref{ex1} and \ref{ej1} show that there are totally geodesic null hypersurfaces where a null section with zero constant surface gravity exists, but one cannot find a null section with nonzero constant surface gravity, and vice-versa.

	The strategy of the proof of Theorem \ref{teoprin} is the following.  Fixed a null section $\xi\in\mathfrak{X}(L)$ and 
	an arbitrary point $x_0$ of $L$, we construct a geodesic null section $\eta$  around a convenient neighborhood of $x_0$ with
	$\eta_{x_0}=\xi_{x_0}$. We prove that the null expansion associated to $\eta$ satisfies a certain differential equation derived from the Raychaudhuri equation. By using the null convergence condition we prove that this expansion must diverge at a finite parameter value, being determined precisely by the so-called future affine length function associated to $\eta$. Then, the rest of hypotheses of the theorem are used to guarantee that the coefficients of the mentioned differential equation can be conveniently extended beyond their initial domain of definition. This allows us to establish an association between the future affine length function of $\eta$ and  the boundary of the region where the solution of the differential equation associated with the inverse of the expansion of $\eta$ vanishes. Finally, a standard application of the Inverse Function Theorem leads to the conclusion that the mentioned boundary  is smooth, and consequently, the same conclusion holds for the affine length function associated to $\eta$, or equivalently, to $\xi$.  The differentiability of the affine length function associated to $\xi$ is enough to get the existence of a null section with nonzero constant surface gravity.

	The rest of this paper is organized as follows. In section \ref{subsec1.1} we recall some basic concepts pertaining to the geometry of null hypersurfaces that will be used throughout the paper. In section \ref{sec2} we introduce the so-called parameter and affine length functions, crucial in the proof of our main result, and study their main properties. In section \ref{keykey} we give the precise definitions of the concepts appearing in the main hypotheses of Theorem \ref{teoprin}. In section \ref{ex} we provide the aforementioned illustrative examples clarifying how those hypotheses affect the results of this paper. We finally give the proof of Theorem \ref{teoprin} in section \ref{fin}.

	\section{Geometry of null hypersurfaces}\label{subsec1.1}
	
	In this section we review  some basic concepts on the geometry of null hypersurfaces that will be used throughout. Since this is fairly standard material we do so only briefly and without proofs, mostly to establish self-contained notation and terminology for the rest of the paper. 
	We assume the reader is familiarized with the basic geometry of Lorentzian manifolds as given in \cite{oneill}.  
	
	An immersion (resp. embedding) $\phi:\Sigma^k \rightarrow M$ is a {\it null immersion} (resp. {\it null embedding}) if the pullback $\phi^{\ast}g$ degenerates everywhere on $\Sigma$. Since we are dealing with Lorentzian signature, this implies that for every $x\in \Sigma$
	$$ \ell_x:=d\phi_x(T_x\Sigma) \cap (d\phi_x(T_x\Sigma))^{\perp}\subset T_{\phi(x)}M $$ 
	is one-dimensional, so any vector in $\ell_x$ is either zero or null, and moreover any $v \in d\phi_x(T_x\Sigma)\setminus \ell_x$ is spacelike. 
	Throughout this paper,  $\zeta\in\mathfrak{X}(M)$   denotes  a timelike future-directed vector field, which allows us to  pick the unique $\xi \in \mathfrak{X}(\Sigma)$ such that 
	\begin{align*}
		d\phi_x(\xi_x) \in \ell _x \\
		g(d\phi_x(\xi_x),\zeta_{\phi(x)})=-1
	\end{align*}
	for all $x\in \Sigma$, which will be a future-directed null section of $\phi$.
	
	Observe that in the codimension one case, i.e., when $k=n-1$, we have 	\begin{equation}\label{uso}
		d\phi_x(T_x\Sigma) = (\ell _x)^{\perp} \quad \forall x\in \Sigma. 
	\end{equation}
	Although much of what we shall present here can be straightforwardly extended to the context of codimension-one null immersions, for simplicity we carry out all our discussion in this paper for a codimension-one null embedding $i:L\hookrightarrow M$. Following standard usage, we shall identify $L$ with its image under $i$ and omit any mention to $i$ in what follows, unless there is no risk of confusion. In particular, given any null section $\xi \in \mathfrak{X}(L)$ of $L$, equation \eqref{uso} is written as 
	\begin{equation}\label{uso2}
		T_xL = (\xi_x)^{\perp}, \quad \forall x\in L.
	\end{equation}
	A simple computation shows that $\nabla_{\xi}\xi \in (T_xL)^{\perp}$ (see, e.g., \cite[Section 3]{Kupeli}); thus, by equation \eqref{uso2} there exists a unique smooth function $\gs_\xi \in C^{\infty}(L)$, called the {\it surface gravity} of $L$ associated with $\xi$, such that 
	\begin{equation}\label{uso3}
		\nabla_\xi \xi = \gs_\xi \cdot \xi.
	\end{equation}
	The latter equation implies that the maximal integral curves of $\xi$ are always {\it null pregeodesics}. Therefore, any null hypersurface $L$ is ruled by a family of null geodesics (not depending on the choice of the null section $\xi$), unique up to affine reparametrizations, called the {\it null geodesic generators} of $L$. In other words,  the integral curves of any null section coincide with the null geodesic generators up to reparametrization. We shall always assume that a given null geodesic generator is inextendible in $L$, i.e., although it may well have a future/past endpoint $p \in M$, we will have $p \notin L$ in that case. 
	
	The {\it null second fundamental form} of $L$ with respect to a fixed null section $\xi\in\mathfrak{X}(L)$ is the symmetric $(0,2)$-tensor on $L$ given by
	$$
	B_\xi(X,Y)=-g(\nabla_X\xi,Y)
	$$
	for all $X,Y\in\mathfrak{X}(L)$.
	The {\it null mean curvature of $L$ with respect to $\xi$} is the trace of $B_\xi$. Concretely, it is the function $H_\xi:L\rightarrow \mathbb{R}$ given at every $x\in L$ by
	$$
	H_\xi(x)=\sum_{i=1}^{n-2}B_\xi(e_i,e_i),
	$$
	where $\{e_1,\ldots,e_{n-2}\}\subset T_xL$ is any orthonormal set of (necessarily spacelike) vectors which together with $\xi_x$ form a basis of $T_xL$. One can easily check that the null mean curvature is independent of the choice of such an orthonormal set.
	The {\it expansion scalar} associated with $\xi$ is $\theta_\xi= - H_\xi$. As the name indicates, the expansion scalar measures the divergence of the integral lines of $\xi$  and thus of the null geodesic generators.

	If $B_\xi=0$, then $L$ is said to be {\it totally geodesic}. This term carries the usual geometric meaning in this context: $L$ is totally geodesic if and only if geodesics with initial velocity vector tangent to $L$ stay in $L$ at least initially, or equivalently, iff the restriction of the Levi-Civita connection $\nabla$ to $L$ defines a connection thereon \cite[Theorem 37]{Kupeli}. For instance, any Killing horizon is totally geodesic. If the null second fundamental form is pure trace, i.e., $B_\xi=\frac{H_\xi}{n-2} g$, then $L$ is said to be {\it totally umbilic}. The null cones in generalized Robertson-Walker spacetimes (and in particular in spacetimes of constant sectional curvature) are examples of totally umbilic null hypersurfaces which are not totally geodesic, \cite{GutOleat}.
	
	A {\it screen distribution} on $L$ is a rank $n-2$ vector subbundle $\s$ of $TL$ whose fibers are everywhere transverse to a null section $\xi$.  Observe that a screen distribution $\s$ yields a direct sum decomposition 
	$$TL= \s \oplus {\rm span} (\xi).$$
	
	Thus, the metric $g$ induces a positive-definite bundle metric $\sigma$ on $\s$ given by $$\sigma(X,Y) = g(X,Y), \quad \forall X,Y\in \Gamma(\s).$$
	Since $B_\xi(\xi,\cdot)=0$, we can regard $B_\xi$ as a symmetric $(0,2)$ tensor on $\s$. In this case, the second fundamental form can also be expressed as $B_\xi(X,Y)=\sigma(A_\xi^*(X),Y)$, where $A_\xi^*:\Gamma(\s)\rightarrow \Gamma(\s)$ is a $C^{\infty}(L)$-linear self-adjoint (with respect to $\sigma$) endomorphism field. It is often convenient to extend $A_\xi^*$ to $TL$ by linearity by defining $A_\xi^*(\xi)=0$. Thus, we still have $A^*_\xi(U)\in\Gamma(\s)$ for all $U\in\mathfrak{X}(L)$, and we can decompose
	\begin{equation}\label{automorpheq}
		\nabla_U\xi=-A_\xi^*(U)-\tau_\xi(U)\xi.
	\end{equation}
	The ($\s$-dependent) one-form $\tau_\xi$ is called a {\it rotation one-form}. Note that $\gs_\xi=-\tau_\xi(\xi)$.
	
	Observe that it is often the case that geometric objects defined on a null hypersurface depend on the chosen null section. Indeed, if we take another null section $\eta=f \xi$, where $f\in C^\infty(L)$ is a nowhere-vanishing function, then 
	\begin{align}
		\tau_{\eta}&=\tau_\xi-\frac{df}{f}, \nonumber \\
		\gs_{\eta}&=f \gs_\xi+df(\xi), \label{eqcambiogs}\\
		B_{\eta}&=fB_\xi, \label{eqcambioB} \\
		H_{\eta}&=fH_\xi. \label{eqcambiocm}
	\end{align}
	However, the last two of these equations show clearly that being totally geodesic or totally umbilic are intrinsic properties of $L$ that do not depend on any choice of null section.

	An important analytic feature of a null hypersurface is the so-called {\it Raychaudhuri equation}, which asserts that 
	\begin{align}\label{eqRay}
		{\rm Ric}(\xi,\xi)=\xi(H_\xi)-\gs_\xi H_\xi-||A^*_\xi||^2.
	\end{align}
	The term
	$||A^*_\xi||^2$ is called  the {\it shear scalar} and it is given by  
	$||A^*_\xi||_x^2=\sum_{i=1}^{n-2}g(A^*_\xi(e_i),A^*_\xi(e_i))$, where $\{e_1,\ldots,e_{n-2}\}$ is an orthonormal basis of $\s_x$ (see for example \cite{Galloway1,Galloway2,Olea1}). It is easy to show that $||A^*_\xi||^2$ only depends on the choice of the null section, not on the screen.
	
	Since the fiberwise endomorphisms $A^*_\xi(x):\s_x\rightarrow \s_x$ are self-adjoint with respect to the inner product $\sigma_x$, they are diagonalizable. Therefore,
	\begin{equation}\label{ineqh}
		\frac{1}{n-2}H_\xi^2\leq ||A_\xi^*||^2.
	\end{equation}
	If, in addition, $(M,g)$ satisfies the \textit{null convergence condition}, i.e.
	\[
	{\rm Ric}(u,u)\geq0\qquad \hbox{for all null vectors $u\in TM$,}
	\]
	then we get the differential inequality 
	\begin{align}\label{ineqRay}
		0\leq \xi(H_\xi)-\gs_\xi H_\xi-\frac{1}{n-2}H_\xi^2.
	\end{align}
	The Raychaudhuri equation has many interesting consequences in the context of null hypersurfaces. For instance, it have been used to characterize null cones \cite{GutOleanc1,GutOleanc2} and to ensure that, under some curvature conditions, a null hypersurface is totally geodesic \cite{Olea1}.

	\section{The final parameter and affine length functions}\label{sec2}
	
	Let a future-directed null section $\xi\in\mathfrak{X}(L)$ be given on the null hypersurface $L$. We proceed to define in this section a couple of special $\xi$-dependent functions on $L$ that will be crucial in the proof of our main result and whose technical properties will be also explored here. 
	
	Denote by $ \Phi_{\xi}:\mathcal{A}_\xi\subset L\times \mathbb{R} \rightarrow L$ the global flow of $\xi$ on $L$, where $\mathcal{A}_\xi$ is its maximal definition domain. Observe that $\mathcal{A}_\xi$ is open in $L\times \mathbb{R}$ and it contains $L\times \{0\}$. Consider also the exponential map $\exp: \mathcal{D}\subset TM \rightarrow M$ of $(M,g)$, where $\mathcal{D}$ is its maximal definition domain. We know that $\mathcal{D}$ is open in the tangent bundle $TM$, contains the image of the zero section of $TM$ and 
	$$t\in [0,1],\; v\in \mathcal{D}\;\; \Rightarrow\;\; t\cdot v \in \mathcal{D}.$$
	Thus, consider the set
	$$\mathcal{B}_\xi:= \{(x,s) \in L\times \mathbb{R}  :  s\cdot \xi_x \in \mathcal{D} \mbox{ and }\exp(s\cdot\xi_x)\in L \}\subset L\times \mathbb{R},
	$$
	which we show below is open, and define the smooth function $\Psi_{\xi}:\mathcal{B}_\xi \rightarrow L$ given by 
	$$\Psi_{\xi}(x,s):=\exp(s\cdot\xi_x), \quad \forall (x,s) \in \mathcal{B}_\xi.$$ 
	
	We now define the {\it final parameter function} $\fpf_\xi^+:L\rightarrow (0,+\infty]$ and the {\it future affine length function} $\alf_\xi^+:L \rightarrow (0,+\infty]$ by
	\begin{align*}
		\fpf_\xi^+(x)&:=\sup\{t>0:  \{x\}\times[0,t)\subset  \mathcal{A}_\xi\}, \\
		\alf_\xi^+(x)&:=\sup\{s>0:  \{x\}\times[0,s)\subset  \mathcal{B}_\xi\}
	\end{align*}
	for all $x\in L$. Note that both functions are indeed strictly positive and either of them can take an infinite value at some point, and maybe even at all points. For example, if $\xi$ is a complete vector field, then $\fpf^+_\xi(x)=\infty$ for all $x\in L$, while $\alf_\xi^+(x) =+\infty$ if and only if the geodesic null generator of $L$ through $x$ is future-complete.
	
	We analogously define the {\it initial parameter function} and {\it past affine length function} $\fpf_\xi^-,
	\alf_\xi^-:L \rightarrow [-\infty,0)$, respectively, by
	\begin{align*}
		\fpf_\xi^-(x)&:=\inf\{t<0: \{x\}\times (t,0]\subset  \mathcal{A}_\xi\}, \\
		\alf_\xi^-(x)&:=\inf\{s<0: \{x\}\times (s,0]\subset  \mathcal{B}_\xi\}
	\end{align*}
	for all $x\in L$.
	It follows immediately from these definitions that the relations $\alf_{\xi}^-(x)=-\alf_{-\xi}^+(x)$ and $\fpf_{\xi}^-(x)=-\fpf_{-\xi}^+(x)$ hold. 
	Observe, in addition, that 
	\begin{align}
		\mathcal{A}_\xi &=\{(x,t)\in L\times\mathbb{R}:\fpf_\xi^-(x)<t<\fpf^+_\xi(x) \} \label{ca1} \\
		\mathcal{B}_\xi &=\{(x,s)\in L\times\mathbb{R}:\alf_\xi^-(x)<s<\alf^+_\xi(x) \}. \label{ca2}
	\end{align}
	
	\begin{lemma} Let $S$ be any metric space, and let $f_-:S\rightarrow [-\infty,0)$ and $f_+:S\rightarrow (0,\infty]$ be two functions on $S$. Let $A\subset \mathbb{R}\times S$ be given by
		$$
		A=\{(t,x)\in \mathbb{R}\times S:f_-(x)<t<f_+(x)\}.
		$$
		Then, $A$ is open (in the product topology) if and only if $f_-$ is upper semicontinuous and $f_+$ is lower semicontinuous on $S$.
	\end{lemma}
	\begin{proof}
		Suppose that $A$ is open. Fix $x_0\in S$ and by way of contradiction suppose for example that $f_+$ is not lower semicontinuous at $x_0$. In that case, pick a sequence $(x_k)_{k\in \mathbb{N}}$ in $S$ with $x_k\rightarrow x_0$ such that $\liminf{f_+(x_k)} =: c< f_+(x_0)$, which implies $(x_0,c) \in A$. In particular, $0\leq c< +\infty$ and we can assume, up to passing to a subsequence, that $f_+(x_k) \rightarrow c$. Since $A$ is open, we have $(x_k,c+\delta) \in A$ for any large enough $k$ and a fixed small enough $\delta>0$. Hence, one obtains 
		$$c+\delta < f_+(x_k)\quad \hbox{for $k$ large enough,}$$
		which in turn implies $c+\delta\leq c$, an absurd.

		Conversely, take $(x_0,t_0)\in A$. Since $f_-(x_0)<t_0<f_+x_0)$ we can pick $\varepsilon>0$ such that $f_-(x_0)<t_0-\varepsilon<t_0+\varepsilon<f_+(x_0)$.  Using that  $f_-$ is upper semicontinuous and $f_+$ is lower semicontinuous, there is a open set $x_0\in U\subset S$ such that $f_-(x)<t_0-\varepsilon<t_0+\varepsilon<f_+(x)$ for all $x\in U$. Therefore, $(x_0,t_0)\in U\times (t_0-\varepsilon,t_0+\varepsilon)\subset A$ and we conclude that $A$ is open.
	\end{proof}

	The functions $\fpf^{\pm}_\xi, \alf^{\pm}_\xi$ are {\it not} continuous in general. However, since $\mathcal{A}_\xi$ is open it follows from the previous lemma that $\fpf^\pm_\xi$ is lower/upper  semicontinuous. We now show  that $\alf^\pm_\xi$ is lower/upper semicontinuous, and thus that $\mathcal{B}_\xi$ is open in $L\times \mathbb{R}$.

	\begin{lemma}\label{lemalsc}
		Let $(M,g)$ be a spacetime, $L\subset M$ a null hypersurface therein and $\xi\in\mathfrak{X}(L)$ a future-directed null section of $L$. Then $\alf^+_\xi$ is lower semicontinuous and $\alf^-_\xi$ is upper semicontinuous. In particular, $\mathcal{B}_\xi$ is an open subset of $L\times\mathbb{R}$.
	\end{lemma}
	\begin{proof}
		We only carry out the proof of the lower semicontinuity of $\alf^+_\xi$, since the corresponding proof for the upper semicontinuity of $\alf^-_\xi$ is entirely analogous. 
		
		Assume that $\alf^+_\xi$ is not lower semicontinuous at $x_0$. Then there is a number $K_0\in\mathbb{R}$ with $K_0<\alf^+_\xi(x_0)$ and a sequence $(x_k)_{k\in \mathbb{N}}\subset L$ which converges to $x_0$ such that $\alf^+_\xi(x_k)\leq K_0<\infty$ for all $k\in\mathbb{N}$. Again up to passing to a subsequence, we can suppose $\alf^+_\xi(x_k)$ converges to some $0\leq T<\alf^+_\xi(x_0)$. By \eqref{ca2} this means that $(x_0,T) \in \mathcal{B}_\xi$, and consequently $T\cdot x_0 \in \mathcal{D}$ and $p:=\exp(T\cdot \xi_{x_0})\in L$. Pick an open neighbourhood $p\in W\subset M$ and a coordinate system $(z_1,\ldots,z_n)$ defined on $W$ such that 
		\begin{equation}
			\label{dyc}
			L\cap W=\{q\in W: z_1(q)=0\}.
		\end{equation}
		Since $\{\alf^+_\xi(x_k)\cdot x_k\}\rightarrow T\cdot x_0\in \mathcal{D}$, then $\alf^+_\xi(x_k)\cdot \xi_{x_k} \in \mathcal{D}$ for $k$ large enough. In other words, the null geodesic $\gamma_k: t\in [0,\alf^+_\xi(x_k)] \mapsto \exp\left( t\cdot \xi_{x_k}\right)\in M$ is well-defined, but $y_k:=\exp(\alf^+_\xi(x_k)\cdot \xi_{x_k})\notin L$. Since $\{y_k\}\rightarrow p\in W$, eventually $y_k \in W$, in which case we have $\gamma_k(t) \in L\cap W$ for all $t$ close enough to $\alf^+_\xi(x_k)$. By continuity (recall \eqref{dyc}), we have
		$
		z_1(y_k)=0$. However, since $y_k \notin L$, we also have $z_1(y_k) \neq 0$, a contradiction.
	\end{proof}

	\begin{example}\label{adhocexample}
		{\rm This simple example illustrates the fact that $\alf^{\pm}_\xi,\fpf^{\pm}_\xi$ are not continuous in general. Let $$(M,g) = (\mathbb{R}^3, -dt^2+ dx^2+dy^2)$$
			and 
			$$L=\{(t,t,y) \in \mathbb{R}^3 \, : \, t,y\in \mathbb{R}\}\setminus \{(t,t,0) \in \mathbb{R}^3 \, : \, t\geq 1\}.$$
			Pick $\xi = \partial_t + \partial_x$. Since this is a geodesic field, i.e., $\nabla _\xi \xi =0$, we have $\alf^+_\xi = \fpf^+_\xi$. Consider the sequence $\{p_k\}$ with  
			$$p_k= (0,0,1/k) \quad k\in \mathbb{N}.$$
			Then $\{p_k\}$ converges to the origin, where $\alf^+_\xi(0,0,0) = \fpf^+_\xi(0,0,0) =1$, but clearly $\alf^+_\xi (p_k)= \fpf^+_\xi(p_k) =+\infty$, i.e, 
			$$\limsup{\alf^+_\xi (p_k)} = \limsup{\fpf^+_\xi(p_k) }=+\infty >1 = \alf^+_\xi(0,0,0) = \fpf^+_\xi(0,0,0). $$
			Thus $\alf^+_\xi,\fpf^+_\xi$ are not upper semicontinuous. By slightly modifying this example we can also show that $\alf^-_\xi,\fpf^-_\xi$ are not lower semicontinuous in general.}
	\end{example}
	\medskip
	
	We now indicate how a change on the null section affects the final parameter and future affine length functions. 
	
	\begin{lemma}\label{cambiofunciones} Let $L$ be a null hypersurface in the spacetime $(M,g)$. Let $\xi\in\mathfrak{X}(L)$ be a future-directed null section and $f\in C^\infty(L)$ a positive function. 
		Then
		\begin{align*}
			\alf_{f\xi}^+(x)=\frac{\alf^+_{\xi}(x)}{f(x)}\quad\hbox{and}\quad
			\fpf_{f\xi}^+(x)=\int_0^{\fpf_\xi^+(x)}\frac{1}{f(\Phi_\xi(x,t))}dt\qquad\hbox{for all $x\in L$.}
		\end{align*}
		%				\item If $f$ is negative, then
		%			\begin{align*}
			%				\alf_{f\xi}^+(x)=\frac{\alf^-_{\xi}(x)}{f(x)}, \qquad
			%				\fpf_{f\xi}^+(x)=\int^{\fpf_\xi^-(x)}_0\frac{1}{f(\Phi_\xi(x,t))}dt\qquad\hbox{for all $x\in %L$.}
			%			\end{align*}
		%for all $x\in L$.
		%	\end{itemize}
\end{lemma}
\begin{proof}
	The derivation of the formula for $\alf^+_{f\xi}$ is straightforward from the definition of the future affine length function. Suppose now $\alpha:\left(\fpf_{\xi}^-(x), \fpf_{\xi}^+(x) \right)\rightarrow L$ is the (maximal) integral curve of $\xi$ with $\alpha(0)=x$. Let $\beta:\left(\fpf_{f\xi}^-(x), \fpf_{f\xi}^+(x) \right)\rightarrow L$ be the integral curve of $f\xi$ such that $\beta(0) =x$, and define the function $a: [0, \fpf_{f\xi}^+(x)) \rightarrow \mathbb{R}$ 
	given by
	$$a(s) = \int^s_0 f\circ \beta(\lambda) d\lambda, \quad \forall s \in [0, \fpf_{f\xi}^+(x)).$$
	In particular, $a' = f\circ \beta > 0$. Thus, $a$ is increasing and $a(0)=0$. Moreover, $a[0, \fpf_{f\xi}^+(x)) = [0,A)$ for some $0<A\leq +\infty$. On the other hand, $a^{-1} : [0,A) \rightarrow [0, \fpf_{f\xi}^+(x))$ is well-defined and increasing. A straightforward computation shows that $\gamma := \beta \circ a^{-1}$ is an integral curve of $\xi$ starting at $x$, and consequently,
	$$[0,A) = [0,\fpf_{\xi}^+(x)) \quad \mbox{and}\quad \gamma \equiv \alpha|_{[0,A)}.$$
	We conclude that $\beta = \alpha \circ a$. Then, we compute 
	
	\begin{eqnarray}\fpf^+_{f\xi}(x) = \int_0^{\fpf^+_{f\xi}(x)}ds &=& \int_0^{\fpf^+_{f\xi}(x)}\frac{a' (s)}{f(\beta(s))}[\equiv 1]ds \nonumber \\
		&=&\int_0^{\fpf^+_{f\xi}(x)}\frac{a' (s)}{f\circ\alpha (a(s))}ds = \int_0^{\fpf^+_{\xi}(x)}\frac{dt}{f\circ\alpha (t)},
	\end{eqnarray}
	where we have made the change of variables $t=a(s)$ at the last step. Since $\alpha(t) = \Phi_\xi(x,t)$, we obtain the desired formula. 
\end{proof}

\begin{remark}\label{rmkrel}
	{\rm Lemma \ref{cambiofunciones} shows that if the future affine length function associated to a given null section on the null hypersurface $L$ is either smooth or finite at some $x\in L$, then these respective properties hold for any other null section. However, for the final parameter function, although the same statement still holds for smoothness, it may fail for finiteness.}
\end{remark}

Since the integral curves of the null section $\xi$ are pregeodesics, the maps $\Phi_\xi$ and $\Psi_\xi$ are related in a nice way as follows. 
\begin{proposition}\label{smallhfunction}
	Let $L$ be a null hypersurface in the spacetime $(M,g)$ and $\xi\in\mathfrak{X}(L)$ a future-directed null section. For each $x\in L$ there exists a unique smooth increasing diffeomorphism
	$$
	h_x:\left(\alf^-_\xi(x),\alf^+_\xi(x)\right)\rightarrow 
	\left(\fpf^-_\xi(x), \fpf^-_\xi(x)\right)
	$$
	such that $h_x(0)=0$, $h'_x(0)=1$ and $\Psi_\xi(x,s)=\Phi_\xi(x,h_x(s))$ for all $s \in \left(\alf^-_\xi(x),\alf^+_\xi(x)\right)$. Indeed, $h_x$ is a maximally extended solution  to the ODE 
	\begin{equation}\label{eqh2}
		\frac{d^2 h_x}{ds^2}(s)+\gs_\xi(\Psi_\xi(x,s)) \left(\frac{d h_x}{ds}(s)\right)^2=0
	\end{equation}
	with the stated initial conditions. Furthermore, the function 
	$$
	h: (x,s) \in \mathcal{B}_\xi \mapsto h_x(s) \in \mathbb{R}
	$$
	is smooth.  
\end{proposition}
\begin{proof}
	We simplify the notation by writing 
	$$I:= \left(\alf^-_\xi(x),\alf^+_\xi(x)\right)\quad\hbox{and}\quad J:= \left(\fpf^-_\xi(x), \fpf^-_\xi(x)\right).$$
	Let $\alpha: s\in I\mapsto \Psi_\xi(x,s) \in L$ the  null geodesic generator with $\alpha'(0) = \xi_x$, and let $\beta: s \in J \mapsto \Phi_\xi(x,s) \in L$ be the  integral curve on $L$ of the null section $\xi$ with $\beta(0) = x$. In particular we have 
	\begin{equation}\label{intaspregeod}
		\beta'' = \left(\kappa_\xi \circ \beta\right)\beta'.
	\end{equation}
	Define 
	$$
	g: s \in J \longmapsto \int_0^s \gs_\xi(\Phi_\xi(x,u))\, du = \int_0^s \left(\gs_\xi\circ \beta\right)(u)\, du\in \mathbb{R}, 
	$$
	and 
	$$
	\ell: s \in J \mapsto  \int_0^se^{g(t)}\, dt \in \mathbb{R}. 
	$$
	Then, $\ell(0) = 0$, $\ell'(0)=1$ and $\ell' = \exp \circ g >0$. By the inverse function theorem $\ell$ is a global increasing diffeomorphism onto its (open) image $K:= \ell(J)\subset \mathbb{R}$. A simple calculation shows that the inverse function
	$$\rho:= \ell^{-1}: K \rightarrow J$$
	is a solution to the IVP 
	$$\rho'' + \left(\kappa_\xi\circ \beta \circ \rho\right)(\rho')^2 =0,\qquad \rho(0)=0, \qquad \rho'(0)=1.$$
	But in that case (cf. \cite[Exercise 3.19]{oneill}) and in face of equation \eqref{intaspregeod} the curve 
	$\gamma: K \rightarrow L$ given by $\gamma:= \beta \circ \rho$ is a null geodesic with $\gamma'(0) = \xi_x$. So, uniqueness of geodesics with the maximal character of $\beta$ on $L$ yields $K\equiv I$ and
	$$\gamma \equiv \alpha. $$
	If we now put $h_x\equiv \rho$ all the claims but the last one follow. To prove the latter, differentiate the equation 
	$$\Psi_\xi(x,s)=\Phi_\xi(x,h_x(s))$$
	with respect to $s$ to get
	\begin{equation}\label{casi}\frac{\partial \Psi_\xi}{\partial s}(x,s) = h_x '(s)\frac{\partial \Phi_\xi}{\partial s}(x,h_x(s)) \equiv h'_x(s)\xi_{\Psi_\xi(x,s)}.\end{equation}
	Now take $\zeta$ a future-directed timelike vector field defined over $L$ such that $g(\xi,\zeta)=-1$. Taking its pointwise scalar product through $g$ with equation \eqref{casi} and integrating we conclude that 
	\begin{equation}\label{eqh}
		h_x(s)=-\int_0^s g(\zeta_{\Psi_\xi(x,u)},\partial \Psi_\xi(x,u)/\partial u) du,
	\end{equation}
	from which the smoothness of $h$ follows. 
\end{proof}

\begin{proposition} \label{lemaed} If $\xi\in\mathfrak{X}(L)$ is a future-directed null section on the null hypersurface $L$, then the following statements hold. 
	\begin{enumerate} 
		\item
		
		Given $x\in L$  and $s\in (\alf^-_\xi(x),\alf^+_\xi(x))$,  if the function $\alf^+_\xi:L\rightarrow (0, \infty]$ is (finite and) smooth on a neighbourhood $\mathcal{V}\ni \Psi_\xi(x,s)$ (in $L$), then it is (finite and) smooth in some neighbourhood $\mathcal{U}\ni x$ (in $L$). 
		
		\item\label{p1prop} If $\alf^+_\xi$ is (finite and) smooth on an open set $\mathcal{U}\subset L$, then 
		\begin{equation}\label{TheEq}
			\xi( \alf^+_\xi)+\gs_\xi\cdot \alf^+_\xi+1=0\quad\text{ on }\quad  \mathcal{U}.
		\end{equation}
		In particular, if we take $\eta=(\alf^+_\xi\cdot \xi)|_{\mathcal{U}}$ then $\gs_\eta(x)=-1$ and $\fpf_\eta^+(x)=\infty$ for all $x\in \mathcal{U}$.

		\item \label{p3prop} If the surface gravity $\gs_\xi$ is constant along the null geodesic generator through a point $x\in L$, then 
		$$
		\alf^+_\xi(x)=
		\begin{cases}
			\frac{e^{\fpf^+_\xi(x)\gs_\xi(x)}-1}{\gs_\xi(x)} & \text{if $\fpf^+_\xi(x)<\infty$ and $\gs_\xi(x)\neq 0$} \\
			-\frac{1}{\gs_\xi(x)} & \text{if  $\fpf^+_\xi(x)=\infty$ and $\gs_\xi(x)<0$} \\
			\infty & \text{if $\fpf^+_\xi(x)=\infty$ and $\gs_\xi(x)\geq 0$.} 
		\end{cases}
		$$
		In particular, if $\xi$ is future complete and $\gs_\xi$ is a negative constant along each null geodesic generator, then $\alf^+$ is  
		(finite and) smooth.

		\item\label{pgprop} If $\alf^-_\xi,\alf^+_\xi:L\rightarrow\mathbb{R}$ are smooth then the null section $\eta=\left(\alf^+_\xi-\alf^-_\xi\right)\cdot \xi$ is geodesic.
	\end{enumerate}
\end{proposition}
\begin{proof} $(1)$ Fix $x\in L$ and $s\in (\alf^-_\xi(x),\alf^+_\xi(x))$. Consider the null geodesic generators 
	$$
	\alpha: s\in (\alf^-_\xi(x),\alf^+_\xi(x)) \rightarrow L \quad\mbox{and} \quad \beta: (\alf^-_\xi(\Psi_\xi(x,s)),\alf^+_\xi(\Psi_\xi(x,s))) \rightarrow L
	$$ 
	parametrized so that $\alpha'(0) = \xi_x$ and $\beta' (0) = \xi_{\Psi_\xi(x,s)}$. Then 
	$\alpha(s)=\Psi_\xi(x,s) = \beta(0)$, and we conclude that $\beta$ is a reparametrization of $\alpha$,
	$$
	\beta(t) = \alpha(A\, t + s),
	$$
	where $A = A(x,s)\neq 0$ is some constant depending on our choices of $x$ and $s$. The previous equation can also be written as
	\begin{align}\label{noseescapa}
		\Psi_\xi\left(x, A(x,s)t+s\right)=\Psi_\xi(\Psi_\xi(x,s),t). 
	\end{align}
	In particular, the identity (\ref{noseescapa}) provides two expressions for the same geodesic with parameter $t$. Looking at the right-hand side of (\ref{noseescapa}) we deduce that the upper limit of the interval of definition of that geodesic is $\alf^+_\xi(\Psi_\xi(x,s))$. So, taking the limit as $t$ goes to $\alf^+_\xi(\Psi_\xi(x,s))$ in the second argument of the left-hand side of (\ref{noseescapa}), we deduce the following relation
	\begin{align}\label{eqaantes}
		A(x,s)\alf^+_\xi(\Psi_\xi(x,s)) + s=\alf^+_\xi(x). 
	\end{align}
	In order to understand a little more closely the dependence of $A$ on $(x,s)$, pick a timelike vector field  $\zeta$ defined over $L$ with $g(\zeta,\xi)=-1$. Differentiating equation \eqref{noseescapa} with respect to $t$, setting $t=0$ and recalling $\frac{\partial \Psi_\xi }{\partial s}(z,0)=\xi_z$ we conclude that 
	$$A(x,s) = -\frac{ 1   }{g(\zeta_{\Psi_\xi(x,s)},\frac{\partial \Psi_\xi}{\partial s}(x,s))},$$
	and hence
	\begin{align}\label{eqa1}
		\alf^+_\xi(\Psi_\xi(x,s))=-g\left(\zeta_{\Psi_\xi(x,s)},\frac{\partial \Psi_\xi}{\partial s}\right)\left(\alf^+_\xi(x)-s\right)
	\end{align}
	Therefore, insofar as $x$ has been arbitrarily chosen, we deduce from \eqref{eqa1} that if $\alf^+_\xi$ is smooth on a neighbourhood of $\Psi_\xi(x,s)$, then it is smooth on some neighbourhood of $x$ as claimed.
	
	\medskip
	
	$(2)$ Assume now that $\alf^+_\xi$ is smooth on some neighborhood $\mathcal{U}\subset L$ and pick again $x \in \mathcal{U}$. If we differentiate equation \eqref{eqa1} with respect to $s$ at $s=0$ and take into account that $\frac{\partial \Psi_\xi}{\partial s}(x,0)=\xi_x$ and $\gs_\xi=g(\nabla_\xi\zeta,\xi)$, then we get 
	\begin{align*}
		\xi(\alf^+_\xi))=-\gs_\xi \alf^+_\xi(x)-1.    
	\end{align*}
	Using this and equation \eqref{eqcambiogs} we see that $\eta=\alf^+_\xi\xi$ has constant surface gravity $\gs_\eta=-1$ on $\mathcal{U}$. Thus, on the one hand, Lemma \ref{cambiofunciones} yields  
	\begin{align*}
		\fpf_\eta^+(x)=\int_0^{\fpf_\xi^+(x)}\frac{1}{\alf^+_\xi(\Phi_\xi(x,t))}dt=\int_0^{\alf_\xi^+(x)}\frac{1}{\alf^+_\xi(\Psi_\xi(x,s))}h'_x(s)ds,
	\end{align*}
	and on the other hand, using equations \eqref{eqh} and \eqref{eqa1}, we obtain
	\begin{align*}
		\fpf_\eta^+(x)=\int_0^{\alf^+_\xi(x)} \frac{1}{\alf^+_\xi(x)-s}ds=\infty.
	\end{align*}
	
	\medskip 
	$(3)$ The  case $\gs_\xi(x)=0$ is trivial, so suppose $\gs_\xi(x)\neq0$. From equation \eqref{eqh2} we get $h_x(s)=\frac{1}{\gs_\xi(x)}\ln \left( \gs_\xi(x)s+1\right)$. Then, the conclusion follows by observing that 
	$$h_x(s) \rightarrow \fpf_\xi^+(x)\;\; \Longleftrightarrow\;\; s \rightarrow \alf^+_\xi(x).$$
	\medskip
	
	$(4)$
	Since $\alf^-_\xi=-\alf^-_{-\xi}$, from equations \eqref{eqcambiogs}, \eqref{TheEq} and Lemma \ref{lemalsc} one deduces that $\alf^-_\xi$ satisfies 
	$$
	\xi(\alf^-_\xi)+\gs_\xi\alf^-_\xi+1=0.
	$$
	Therefore, $\xi( \alf^+_\xi-\alf^-_\xi)+\gs_\xi(\alf^+_\xi-\alf^-_\xi)=0$, which means that $\eta=(\alf^+_\xi-\alf^-_\xi)\xi$ is geodesic, as desired.
\end{proof}

\begin{proposition}\label{propunico}
	Let $\xi_1,\xi_2\in\mathfrak{X}(L)$ be two null sections in a null hypersurface $L$ with identical nonzero constant surface gravity. If $\xi_1$ and $\xi_2$ are complete then $\xi_1=\xi_2$. 
	
	In particular, if $L$ is compact and admits a null section with nonzero constant surface gravity, then this section is unique up to homothety.   
	
\end{proposition} 
\begin{proof}
	We can assume that both $\xi_1$, $\xi_2$ are future-directed,
	$\gs_{\xi_1}= -1$ and $\gs_{\xi_2}=\pm 1$.
	If we suppose that $\gs_{\xi_2}=1$, since  $b_{\xi_1}^+(x)=b_{\xi_2}^+(x)=\infty$, from Proposition \ref{lemaed} (\ref{p3prop}) we get 
	$\alf^+_{\xi_1}=1$ and $\alf^+_{\xi_2}=\infty$, in contradiction with Lemma \ref{cambiofunciones}.
	Therefore, $\gs_{\xi_2}=-1$ and again by  Proposition \ref{lemaed} (\ref{p3prop}) and Lemma \ref{cambiofunciones} we conclude that  $\xi_1=\xi_2$.	
\end{proof}

The following simple consequence of the Raychaudhuri equation will be relevant later. 
\begin{proposition}\label{lemaNCCfinita} Let $(M,g)$ be a spacetime which obeys the null convergence condition, $L\subset M$ a null hypersurface and  $\xi\in\mathfrak{X}(L)$  a future-directed null section. 
	If $H_\xi(x)>0$ for some $x\in L$ then  $\alf^+_\xi(x)\leq\frac{n-2}{H_\xi(x)}(<\infty)$.
\end{proposition}
\begin{proof} Consider the null geodesic 
	$\alpha:[0,\alf^+(x))\rightarrow L$ given by
	$\alpha(s)=\exp(s\,\xi_x)=\Psi_\xi(x,s)$. Since the PDE \eqref{basicpde} always admits {\it local} solutions, and in particular for $c=0$, given $t_0\in (0,\alf^+(x))$, there exist an open interval $t_0\in I\subset\mathbb{R}$, an open set $\alpha(t_0)\in U\subset L$ and a geodesic null section $\eta_I\in\mathfrak{X}(U)$ such that $\eta_{I}(\alpha(s))=\alpha'(s)$ for all $s\in I$. Since $\alpha$ is a geodesic, equation \eqref{ineqRay} implies the function $y:I\rightarrow\mathbb{R}$ given by  $y(s)=H_{\eta_I}(\alpha(s))$ satisfies the inequality 
	\begin{align}
		0\leq y'(s)-\frac{1}{n-2}y(s)^2.\label{ineq1}
	\end{align}
	
	For any two overlapping intervals $I,J$ as before, since $\eta_I(\alpha(s))=\eta_J(\alpha(s))=\alpha'(s)$ for all $s\in I\cap J$, then from equation \eqref{eqcambiocm} we have $H_{\eta_I}(\alpha(s))=H_{\eta_J}(\alpha(s))$ for all $s\in I\cap J$.
	Therefore, we can construct a well-defined smooth function $y:[0,\alf^+(x))\rightarrow \mathbb{R}$ satisfying the inequality \eqref{ineq1}. Since $\alpha'(0)=\xi_x$, we also have $y(0)=H_\xi(x)$.
	
	Finally, since $y(0)>0$, the inequality \eqref{ineq1} implies that $y(s)>0$  and 
	$0<\frac{1}{y(s)}\leq \frac{1}{y(0)}-\frac{1}{n-2}s$ for all $s\in[0,\alf^+(x))$. Thus $\alf^+(x)\leq\frac{n-2}{y(0)}=\frac{n-2}{H_\xi(x)}$. 
\end{proof}

\section{Key geometric conditions on null hypersurfaces}\label{keykey}
We wish to give here the detailed definitions of the conditions used in, and briefly described after, Theorem \ref{teoprin}. 

\begin{definition}[Future prolonged generators]\label{def.1.1}
	A (future-directed) null geodesic generator $\gamma:(a,b) \rightarrow L$ ($-\infty \leq a<b\leq +\infty$) of the null hypersurface $L$ is said to be {\em future prolonged} if it is either future complete (i.e., $b=\infty$) or else  it is future extendible in some isometric extension $(\tilde{M},\tilde{g})$ of $(M,g)$ (which may depend on the generator) of $(M,g)$. In particular, there exists a future null geodesic $\tilde{\gamma}:(a,\tilde{b})\rightarrow \tilde{M}$ with $b<\tilde{b}\leq \infty$, $\tilde{\gamma}|_{(a,b)} = \gamma$ and $\tilde{\gamma}(b) \in \tilde{M}\setminus L$. 
\end{definition}
{\it Past prolonged} null geodesic generators are defined in a time-dual fashion. Henceforth, definitions/results will often be given only in ``future versions'' and the obviously obtained ``past versions'' will be then understood.

Evidently, the only way a null hypersurface will {\it not} have future prolonged null generators is when they are future inextendible but still future incomplete in {\it any} isometric extension of $(M,g)$. (This could occur, for example, due to a blow-up of curvature along one or all the null generators.)

\begin{definition}[Future weak inextendibility]\label{def.1.2}
	Given a future null geodesic generator $\gamma:(a,b) \rightarrow L$ ($-\infty \leq a<b\leq +\infty$) of the null hypersurface $L$, we say that $L$ is {\em future weakly extendible along $\gamma$} if there exists a null codimension one immersion $\phi:\Sigma ^{n-1} \rightarrow M$ such that $L\subset \phi(\Sigma)$ and $\gamma$ admits a future endpoint $p \in \phi(\Sigma)\setminus L$. We thus say that $L$ is {\em future weakly inextendible} if it is not future weakly extendible along any of its null geodesic generators. 
\end{definition}
Observe that future weak inextendibility of the null hypersurface $L$ may hold and yet $L$ be extendible. A simple but instructive example is given in the spacetime $(M,g) = (\mathbb{R}^3, -dt^2 + dx^2 +dy^2)$ (with time-orientation so that $\partial_t$ is future-directed) by the ``null cone without two generators''
$$L= \{(t,x,y) \in \mathbb{R}^3 \, : \, t^2=x^2+y^2, \, t<0, \, y\neq 0\}.$$
Thus, $L$ is clearly extendible, say via the embedded cone $L'\supset L$ given by 
$$L'= \{(t,x,y) \in \mathbb{R}^3 \, : \, t^2=x^2+y^2,\, t<0\}.$$
Both $L$ or $L'$ are future weakly inextendible in the above sense, but only $L'$ is truly inextendible. In addition, both $L,L'$ have future prolonged generators, but if one deletes the origin $(0,0,0)$ from the spacetime, then $L,L'$ would {\it not} have future prolonged generators, while remaining future weakly inextendible in the resulting spacetime.

We end this section with a key definition, which together with Definitions \ref{def.1.1} and \ref{def.1.2}, precisely identifies the class of null hypersurfaces we shall be interested in as per our main results.

\begin{definition}[Future finite type]\label{def.1.3} Let $L\subset M$ be a null hypersurface in the spacetime $(M,g)$ and  $\xi\in\mathfrak{X}(L)$ a future-directed null section thereon.
	Consider the map $$\Gamma_\xi:\mathcal{B}_\xi-\{(x,s) \in \mathcal{B}_\xi: H_\xi(\Psi_\xi(x,s))=0\}\rightarrow \mathbb{R},$$ given by
	$$
	\Gamma^+_\xi(x,s):=\frac{||A^*_\xi||^2_{\Psi_\xi(x,s)}}{H^2_\xi(\Psi_\xi(x,s))},
	$$
	where $H_\xi$ is the null mean curvature function and $A^*_\xi$ denotes the endomorphism defined around \eqref{automorpheq}. We say that 
	$L$ is of {\em future finite type} if:
	\begin{itemize} \item For each $x\in L$, we have that $H_\xi(\Psi_\xi(x,s))$ is positive (i.e., the expansion is negative) for some $s\in \left(\alf^-_\xi(x),\alf^+_\xi(x)\right)$.
		\item  
		$\Gamma^+_\xi$ can be extended to a smooth map $\tilde{\Gamma}^+_\xi$ defined in a open set $\mathcal{O}\subset L\times\mathbb{R}$ such that $(x,\alf^+_\xi(x))\in\mathcal{O}$ for all $x\in L$ with $\alf^+_\xi(x)<\infty$.
	\end{itemize}
\end{definition}

The above definition does not depend on the chosen null section insofar as it is future-directed. Indeed,  take $f\,\xi\in\mathfrak{X}(L)$  another future-directed null section, with $f\in C^\infty(L)$ a positive function, and note that $\Psi_{f\, \xi}(x,s)=\exp(sf(x)\, \xi_x)=\Psi_\xi(x,f(x)s)$. Using  equations \eqref{eqcambioB} and \eqref{eqcambiocm} we have that 
%\marginpar{Cambiar $\Psi$ por $\Psi_\xi(x,s)$}
\[
\begin{array}{c}
	\Gamma^+_{f\xi}(x,s)=
	\frac{||A^*_{f\xi}||^2_{\Psi_{f\, \xi}(x,s)}}{H_{f\xi}(\Psi_{f\, \xi}(x,s))^2} \qquad\qquad\qquad\qquad\qquad\qquad\qquad\qquad\qquad\qquad\qquad\quad \\  \quad\qquad\qquad\qquad\qquad =\frac{||A^*_{f\xi}||^2_{\Psi_{f\, \xi}(x,s)} }{H_{f\xi}(\Psi_\xi(x,f(x)s))^2}=
	\frac{||A^*_\xi||^2_{\Psi(x,f(x)s)}}{H_\xi(\Psi_\xi(x,f(x)s))^2}
	=\Gamma^+_\xi(x,f(x)s).
\end{array}
\]
Hence, $\Gamma^+_{\xi}$ can be extended to a smooth map which contains in its domain the points $(x,\alf^+_{\xi}(x))$ for all $x\in L$ iff
$\Gamma^+_{f\xi}$ can be extended to a smooth map which contains in its domain the points  
$$
(x,\alf^+_{f\xi}(x))=\left(x,\alf^+_\xi(x)/f(x)\right)\quad\hbox{for all $x\in L$.}
$$
Observe that inequality \eqref{ineqh} implies $\Gamma^+_\xi(x,s)\geq \frac{1}{n-2}$ for all $(x,s)$ on its domain. On the other hand, if $L$ is totally umbilic with everywhere-negative expansion, then $\Gamma^+_\xi(x,s)$ is {\it constantly equal} to $\frac{1}{n-2}$, and consequently, {\it any totally umbilic null hypersurface  with everywhere-negative expansion is of finite type}.

The conditions in Definitions \ref{def.1.1}, \ref{def.1.2} and \ref{def.1.3} evoke the behavior of a null cone in Minkowski spacetime near its vertex. However, even in the Minkowski spacetime, there are other examples of null hypersurfaces that satisfy these three conditions, apart from null cones. Consider, for instance, the null hypersurface 
$$
L=\{(x,y,z,t)\in\mathbb{L}^4: t>0, t=\sqrt{x^2+y^2}\}.
$$
Pick as null section $\xi=-x\partial_x-y\partial_y-t\partial_t$, which we assume is future-directed. Pick a screen distribution given by $\mathcal{S}={\rm span}\{y\partial x-x\partial y, \partial z\}$. Since 
\begin{align*}
	\nabla_{\partial_z}\xi&=0,\\
	\nabla_{y\partial_x-x\partial_y}\xi&=y\partial_x-x\partial_y,
\end{align*}
the corresponding principal curvatures of $L$ with respect to $\xi$ are constantly equal to $0$ and $1$, resp. Therefore, $L$ is of future finite type. It is clear that $L$ is future weakly inextendible and has future prolonged null generators.

\section{Some noteworthy examples}\label{ex}

Our goal in this section is to provide a number of illustrative examples clarifying the scope of the results in this paper as well as of other similar ones in the literature.

We begin by emphasizing that {\it a null section with constant surface gravity may not exist}. In addition, whether such a putative constant is zero or nonzero can affect existence. 

As a simple example of the latter phenomenon, in which a globally defined null section with {\it nonzero} constant surface gravity does not exist, is a degenerate hyperplane in Minkowski spacetime. In fact, since $\alf^\pm_\xi=\pm\infty$  in this case, this directly follows from  Proposition \ref{lemaed}(\ref{p3prop})\footnote{More generally, and by the same token, no null section with nonzero surface gravity will exist on a null hypersurface with complete null geodesic generators.}. Of course, a {\it geodesic} null section, hence with identically zero constant gravity, can be defined in this case. 

The previous example is very admittedly special. But one can obtain many further examples of the same kind via the following result.  

\begin{proposition}\label{propnone} Let $L\subset M$ be a null hypersurface in the spacetime $(M,g)$ and $\xi\in\mathfrak{X}(L)$ a future-directed null section on $L$. Fix $x_0\in L$ and  take $\gamma:I\rightarrow L$ the integral curve of $\xi$ through $x_0$. Suppose that:
	\begin{itemize}
		\item 
		$\fpf^+_\xi(x_0)=\infty$ and $\fpf^-_\xi(x_0)=-\infty$.
		\item $\kappa_\xi(\gamma(t))\geq 0$ for $t\geq 0$ and $\kappa_\xi(\gamma(t))\leq 0$ for $t\leq 0$.
	\end{itemize}
	Then $L$ does not admit a null section with nonzero constant surface gravity. Moreover,
	the null geodesic generator of $L$ through $x_0$ is complete, i.e., $\alf^+_\xi(x_0)=\infty$ and $\alf^-_\xi(x_0)=-\infty$.
\end{proposition}
\begin{proof}
	Assume the existence of a null section $\eta = f\, \xi\in\mathfrak{X}(L)$ with constant surface gravity  $c\neq0$. Without loss of generality, we can suppose that $c=1$. From equation \eqref{eqcambiogs},  the function $y(t)=f(\gamma(t))$ satisfies 
	$$
	y'(t)=1-\gs_\xi(\gamma(t)) y(t). 
	$$
	Now, $f$ is nowhere zero, and hence either strictly positive or strictly negative on the connected image of $\gamma$. If $y=f\circ \gamma <0$ then $y'(t)\geq 1$ for $t\geq 0$ and thus $t+y(0)\leq y(t)< 0$ for positive $t$, and we get a contradiction by evaluating at $t=-y(0)>0$.
	If $f\circ \gamma>0$, then we get a contradiction in an analogous way.
	
	Finally, using the same notation as in Proposition \ref{smallhfunction}, we have $\gamma(h_{x_0}(s)) = \Phi_\xi(x_0,h_{x_0}(s)) = \Psi_\xi(x_0,s)$, and $$h_{x_0}(s)\geq 0 \Longleftrightarrow 0\leq s < \alf^+_\xi(x_0).$$
	From equation \eqref{eqh2} in Proposition \ref{smallhfunction} we have $h''_{x_0}(s)\leq 0$ for all $0\leq s<\alf^+_\xi(x_0)$, and thus (recall that $h_{x_0}(0)=0$ and $h'_{x_0}(0)=1$) $h_{x_0}(s)\leq s$ for all $0\leq s<\alf^+_\xi(x_0)$, whence $\fpf^+_\xi(x_0) \leq \alf^+_\xi(x_0)$. Since $\fpf^+_\xi(x_0)=\infty$, we get $\alf^+_\xi(x_0)=\infty$. Analogously, we can show that  $\alf^-_\xi(x_0)=-\infty$. 
\end{proof}

Observe that the second condition in the previous proposition does not depend on the time orientation of the null section.

In order to illustrate the broad applicability of Proposition \ref{propnone}, we consider a\-no\-ther class of examples where no null section with nonzero constant surface gravity exists, but which do possess a geodesic null section.
\begin{example}\label{ex1} {\em Take $M=\mathbb{R}^4$ with the Lorentz metric given in global Cartesian coordinates $(s,x,y,z)$ by
		$$
		g=2s\phi(x,y,z) \left(dx^2-ds^2\right) -2dsdx+dy^2+dz^2,
		$$
		where $\phi:\mathbb{R}^3\rightarrow\mathbb{R}$ is an arbitrary smooth function, and the time orientation is chosen such that the timelike  vector field $\zeta=\partial_x+ \partial_s$ is future-directed. Here, $L=\{(0,x,y,z):x,y,z\in\mathbb{R}\}$ is a totally geodesic null hypersuface with future-directed null section $\xi=\partial_x|_L$. 
		We can compute its surface gravity as follows. Since $g(\zeta,\xi)=-1$, we have
		$$
		\gs_\xi=g(\xi,\nabla_{\xi}\zeta)=g(\partial_x,\nabla_{\partial_x}\partial_s)=\frac{1}{2} \partial_s\, g(\partial_x,\partial_x)=\phi.
		$$
		Consider the integral curve of $\xi$ given by $\gamma(t)=(0,t,0,0)$ for all $t\in\mathbb{R}$. Since $\gs_\xi(\gamma(t))=\phi(t,0,0)$, a suitable choice of $\phi$ (for instance $\phi(x,y,z) = x$) together with an application of Proposition \ref{propnone} allows us to conclude that a null section with nonzero constant surface gravity does not exist on $L$. 
		However, note that from equation \eqref{eqcambiogs} the null section $\eta=f\, \partial_x$, with $f(x,y,z)=e^{-\int_0^x \phi(r,y,z)dr}$, is a geodesic null section.}
\end{example}

The following example shows that the completeness hypothesis for the null sections cannot be dropped in Proposition \ref{propunico}. It also illustrates the fact that we can perfectly well have {\it both} a null section with nonzero constant surface gravity {\it and} a geodesic null section on the same (weakly inextendible) null hypersurface.

\begin{example}\label{ejemplocono}
	{\em Consider the null cone in the $n$-dimensional Minkowski spacetime $\mathbb{L}^n$ given by
		$$
		L=\{(x_0,\ldots,x_{n-1})\in \mathbb{L}^n: x_0^2=\sum_{k=1}^{n-1}x_k^2,x_0>0\}\subset \mathbb{L}^n.
		$$
		We of course know that $L$ is a null hypersurface and $\xi:=\sum_{k=0}^{n-1}x_k \partial_{x_k}$ is a complete null section with surface gravity $\gs_{\xi}=1$. Moreover, $\eta=f\xi$ with $f=\frac{1}{x_0}$ is geodesic. Therefore, Equation\eqref{eqcambiogs} implies that the function $f$ satisfies 
		$$
		0=f+\xi(f).
		$$
		We now consider the (noncomplete) null section $\xi':=\xi+\eta=\left(1+f\right)\xi$.  Since $$1=(1+f)+\xi(1+f),$$
		$\xi'$ also has constant surface gravity $\gs_{\xi'}=1$.

		The previous discussion points to a general behavior of null cones in {\it any} spacetime $(M,g)$. More precisely, given $p\in M$ and $e_0\in T_pM$ a future-directed unit timelike vector, let
		$$
		C^+_{p}:=\{exp_p(u):u\in \mathcal{D}_p, g(u,u)=0,g(u,e_0)=-1\},
		$$
		where $\mathcal{D}_p = \mathcal{D}\cap T_pM$ is the maximal domain of $\exp_p$, be the latter of the future null cone in $T_pM$. Near the vertex, $C^+_{p}$ is a null hypersurface: if we pick a normal neighborhood $\mathcal{U}$ of $p$, then $$S_p:= \{\exp_p(u):u\in\exp_p^{-1}(\mathcal{U}),g(u,u)=0,g(u,e_0)=-1\}$$
		is a null hypersurface contained in $C^+_{p}$. If we choose suitable normal coordinates $(x_0,\ldots,x_{n-1})$ around $p$ with $\partial_{x_0}|_{p}=e_0$, then $\xi=\sum_{i=0}^{n-1} x_i\partial_{x_i}|_S$ is a null section with constant surface gravity $1$, whereas $\xi'=\frac{1}{x_0}\xi$ is a geodesic null section of $S$.
		Therefore, as before, $\xi'=\xi+\eta$ is another null section with constant surface gravity $\gs_{\xi'}=1$.}
	
\end{example}

Next, we provide an example with no geodesic null sections, but admitting a null section with constant surface gravity equal to $-1$.

\begin{example}\label{ej1} {\em With the conventions and notation of Example \ref{ex1} we again take $\widetilde{M}=\mathbb{R}^4$ with metric 
		$$
		\widetilde{g}=2s\phi(y,z)\left(dx^2-ds^2\right)-2dsdx+dy^2+dz^2,
		$$
		where now the smooth function $\phi$ is assumed negative and satisfying 
		
		$$\phi(y+2\pi n,z+2\pi m)=\phi(y,z)
		$$
		for all $y,z\in\mathbb{R}$ and $n,m \in \mathbb{Z}$. As in Example \ref{ex1}, $\widetilde{L}=\{(0,x,y,z):x,y,z\in\mathbb{R}\}$ is a totally geodesic null hypersurface and $\xi=\partial_x$ is a null section with surface gravity $\gs_\xi=\phi$. Fixed $(0,x,y,z)\in L$, we can easily check that 
		$$
		\alpha(s)=\left(0,x+\mathfrak{h}(s),y,z\right),\quad \hbox{with $\;\mathfrak{h}(s)= \frac{1}{\phi(y,z)}\ln(1+\phi(y,z)s)$,}
		$$
		is a future-incomplete null geodesic generator in $\widetilde{L}$ with $\alpha(0)=(0,x,y,z)$ and $\alpha'(0)=\xi_{\alpha(0)}$. Observe that the features we have imposed on $\phi$ ensure the existence of an isometric, properly discontinuous action of $\mathbb{Z}^3$ preserving $L$ given by $(s,x,y,z) \mapsto (s,x+2\pi m,y+2\pi  n,z+2\pi k)$ $m,n,k \in \mathbb{Z}$. The quotient $M = \mathbb{R}^4/\mathbb{Z}^3 = \mathbb{R}\times \mathbb{S}^1\times \mathbb{S}^1\times \mathbb{S}^1$ inherits a unique metric $g$ and time orientation so that the canonical projection $\pi:\widetilde{M} \rightarrow M$ becomes a time-orientation-preserving Lorentzian covering map. Then $L=\pi(\widetilde{L})$ is a compact totally geodesic null hypersurface in $(M,g)$. The geodesic $\alpha$ projects down as null incomplete geodesic generator of $L$.
		
		Looking at the domain of definition of $\mathfrak{h}$ we infer that 
		$$
		\alf^+_{\xi}(0,x,y,z)=-\frac{1}{\phi(y,z)}\quad\hbox{on $L$},$$ where we understand that the appropriate identifications on the coordinates hold and that $\xi$ is now the well-defined projection of the null section $\xi$ to $L$, which we still denote in the same way. $\alf^+_\xi$ is a smooth function, and Proposition \ref{lemaed} guarantees that $-\frac{1}{\phi}\xi$ is a complete null section with constant surface gravity $-1$. Furthermore, Proposition \ref{propunico} entails that this is the unique null section on $L$ with this property. But no geodesic null section exists, since any such would have to be complete by the compactness of $L$, but would have the projection of the incomplete null geodesic $\alpha$ as an integral curve, which is a contradiction.}
\end{example}

The following final example shows that a null hypersurface may not admit sections with constant (either zero or nonzero) surface gravity.

\begin{example}\label{ej2}
	{\em If we take $\phi(y,z)=\cos y$ in the previous example, we obtain via the quotient taken therein a compact totally geodesic null hypersurface $L$ with affine length functions for $\xi=\partial_x$ given by
		\[
		\alf^+_{\xi}(0,x,y,z)=\begin{cases} \infty & \text{if $\cos y\geq 0$} \\ \frac{-1}{\cos y} & \text{if $\cos y<0$}\end{cases}.\qquad
		\alf^-_{\xi}(0,x,y,z)=\begin{cases} \frac{-1}{\cos y} & \text{if $\cos y>0$} \\
			\infty & \text{if $\cos y\leq 0$.} 
		\end{cases}
		\]
		Reasoning as in previous example, there are no geodesic null sections on $L$. Furthermore, no null section with nonzero constant surface gravity exists either. Indeed, suppose by contradiction that $\xi'$ is such a null section with, say, $\gs_{\xi'}=-1$. By Proposition \ref{lemaed} (\ref{p3prop}) we have $\alf^+_{\xi'}(0,x,y,z)=1$ for all $(0,x,y,z)\in L$,  and thus by Lemma \ref{cambiofunciones} it follows that $\alf^+_\xi(0,x,y,z)$ would be finite for all $(0,x,y,z)\in L$, which is contradiction.}
\end{example}

\section{Proof of the main theorem}\label{fin}

We now turn to the proof of Theorem \ref{teoprin}, which will be carried out via a number of subsidiary technical results.

The first of these has perhaps independent geometric interest, as it gives conditions to ensure an open strip exists around a large enough piece of a null geodesic generator. A result of this sort was proven 
in \cite[Proposition 22]{Kupeli} under the assumption of strong causality of $(M,g)$.

\begin{proposition}\label{propbanda1}
	Let $L$ be a null hypersurface in the spacetime $(M,g)$ and $\xi\in\mathfrak{X}(L)$ a future-directed null section. Fix  a point  $x_0\in L$ and assume that at least one of the following conditions holds:
	\begin{enumerate}
		\item $(M,g)$ is causal, satisfies the null convergence condition and $H_\xi(x_0)>0$;
		\item $(M,g)$ satisfies the null convergence condition, $H_\xi(x_0)>0$ and $L$ has future prolonged null geodesic generators.
	\end{enumerate}
	Then, there is some open subset    
	$x_0\in V\subset L$ and a geodesic null section $\eta\in\mathfrak{X}(V)$ with $\eta_{x_0}=\xi_{x_0}$ such that $\exp(s\cdot \eta_x)\in V$ for any $x\in V$ and any $0\leq s<\alf^+_\eta(x)$.
\end{proposition}
\begin{proof}Pick
	any ``small'' codimension two (in $(M,g)$) spacelike submanifold $\Sigma\subset L$ containing $x_0$ and transverse to $\xi$. Then, for $\epsilon>0$ small enough,
	$$
	\mathcal{B}_\xi^\Sigma:=\{ (x,s)\in \mathcal{B}_\xi\cap (\Sigma\times\mathbb{R}):  -\varepsilon<s<\alf^+_\xi(x) \}
	$$
	is an open subset of $\Sigma\times\mathbb{R}$ (recall that we know from Lemma \ref{lemalsc} that  $\mathcal{B}_\xi$ is open in $L\times\mathbb{R}$).
	Denote by $\sigma:\mathcal{B}_\xi^\Sigma\rightarrow L$ the restriction\footnote{If there is no risk of confusion, we shall hereafter often drop explicit subscripts '$\eta,\xi$' from $\Psi,\Phi$ etc. in order to simplify the notation.} of $\Psi$ to $\mathcal{B}_\xi^\Sigma$. We first wish to show that $\sigma$ is a local diffeomorphism. Suppose by contradiction that $(x_1,s_1)\in\mathcal{B}_\xi^\Sigma$ is a singular point of $\sigma$. Since $\Psi(x,s)=\Phi(x,h_x(s))$ and $h'_x(s)>0$, then $(x_1,t_1)=(x_1,h_{x_1}(s_1))$ is a singular point of $\Phi$. So, in order to arrive at a contradiction, it suffices to show that $\Phi$ can not have singular points. Suppose that there is a nonzero $v\in T_{x_1}\Sigma$ and $a\in\mathbb{R}$ such that the derivative $\Phi_{*_{(x_1,t_1)}}(v+a\partial _t)=0$ (where $\partial_t$ denotes the standard lift to $\Sigma \times \mathbb{R}$ of the canonical vector field $d/dt$ on $\mathbb{R}$). But  
	\begin{equation}\begin{array}{rl}
			\Phi_{*_{(x_1,t_1)}}(v+a\partial_t) & =\left(\Phi_{t_1}\right)_{*_{x_1}}(v)+a\xi_{\Phi_{x_1}(t_1)} \\ & = \left(\Phi_{t_1}\right)_{*_{x_1}}(v)+a\left(\Phi_{t_1}\right)_{*_{x_1}}(\xi_{x_1})
			=\left(\Phi_{t_1}\right)_{*_{x_1}}(v+a\xi_{x_1}).
		\end{array}
	\end{equation}
	Since $\Phi_{t_1}$ is a local diffeomorphism and $v$ is transverse to $\xi_{x_1}$, we have $v=a=0$, an absurd. So, $\sigma:\mathcal{B}_\xi^\Sigma\rightarrow L$ is a local diffeomorphism.

	Finally, it suffices to prove the injectivity of $\sigma$ (for ``small enough'' $\Sigma$), since, in that case, the open subset $V=\sigma(\mathcal{B}_\xi^\Sigma)\subset L$ and the null section $\eta_{\sigma(x,t)}=\frac{\partial \Psi}{\partial t}(x,t)$ have the required properties. It is here that at least one of the stated assumptions will come into play. 
	
	Suppose by contradiction that $\sigma$ is not one-to-one, independently of how one shrinks $\Sigma$. Then there exist two sequences $\{x_n\}, \{x^*_n\}\in \Sigma$ approaching $x_0$ and $s_n\in(-\varepsilon,\alf_\xi^+(x_n))$, $s^*_n\in(-\varepsilon,\alf_\xi^+(x^*_n))$ such that $(x_n,s_n)\neq(x^*_n,s^*_n)$
	and $\Psi(x_n,s_n)=\Psi(x^*_n,s^*_n)$ for all $n\in \mathbb{N}$.

	If we consider $t_n:=h_{x_n}(s_n)$ and $t^*_n:=h_{x^*_n}(s^*_n)$, then 
	$\Phi_{t_n}(x_n)=\Phi_{t^*_n}(x^*_n)$. Since two integral curves of $\xi$ cannot cross, we conclude that $x_n$ and $x^\ast_n$ must be on the same integral curve, or equivalently 
	\begin{equation}\label{eqphi}
		x^*_n=\Phi_{t_n-t^*_n}(x_n).
	\end{equation}

	Now, assume the properties in hypothesis (1). 
	Observe that causality implies that $x_n\neq x^*_n$ for all $n$.
	By shrinking $\Sigma$ if necessary we can suppose that a number $c>0$ exists such that $c\leq H_\xi(x)$ for all $x\in \Sigma$. By applying Proposition \ref{lemaNCCfinita} we 
	have
	\begin{align*} 
		-\varepsilon< s_n< \alf^+_\xi(x_n)\leq \frac{n-2}{H_\xi(x_n)}\leq \frac{n-2}{c},\\
		- \varepsilon < s^*_n <\alf^+_\xi(x^*_n)\leq \frac{n-2}{H_\xi(x^*_n)}\leq \frac{n-2}{c}.
	\end{align*}
	Thus, up to passing to subsequences we can suppose that $\{s_n\}$ and $\{s^*_n\}$ converge to some $s_0,s^*_0\in\mathbb{R}$.  
	
	If we again let $t_n:=h_{x_n}(s_n)$ and $t^*_n:=h_{x^*_n}(s^*_n)$, then $t_n - t^*_n \rightarrow t_0$, and 
	taking limits in equation \eqref{eqphi} we get $x_0=\Phi_{t_0}(x_0)$. 
	Since $\sigma$ is a local diffeomorphism, it is injective in a neighbourhood of $(x_0,0)$. Therefore, $t_0\neq 0$.
	But in that case the integral curve through $x_0$ is periodic, contradicting the causality of $(M,g)$.
	
	Finally, if we assume hypothesis (2), just note that any periodic integral curve in $L$ through $x_n$ or $x^*_n$ would imply that $\alf^+_\xi(x_n)=\infty$ or $\alf^+_\xi(x^*_n)=\infty$ - contradicting the finiteness for these deduced in the previous case (this part does not use causality!). Now, the lower semicontinuity of $\alf^+_\xi$ also implies finiteness of $\alf^+_\xi(x_0)$, which in turn prevents the integral curve through $x_0$ of being periodic. Thus, a simple adaptation of the argument for the previous case goes through, completing the proof. 
\end{proof}

\begin{proposition}\label{propinfty}
	Let $(M,g)$ be a spacetime, and let $L$ be a null hypersurface therein of future finite type and with future prolonged null geodesic generators. Fix $x_0\in L$ and let $\eta\in\mathfrak{X}(L)$ be a future-directed geodesic null section 
	such that $L$ is weakly future inextendible along the null geodesic generator (integral curve of $\eta$) through $x_0$. If  $\alf^+_\eta(x_0)<\infty$ and $H_\eta(x_0)>0$, then  $$\lim_{s\rightarrow \alf^+_\eta(x_0)}H_\eta(\exp_{x_0}(s\eta_{x_0}))=\infty.$$
\end{proposition}
\begin{proof} 
	Let $\Sigma\subset L$ be a codimension one (in $L$) spacelike submanifold  with $x_0\in \Sigma$.
	Since $\alf^+_{\eta}(x_0)<\infty$ and $L$ has future prolonged generators, there is $\delta>0$ such that
	the map
	$\widetilde{\Psi}:\Sigma\times (-\delta,\alf^+_\eta(x_0)+\delta)\rightarrow M$ given by
	$\widetilde{\Psi}(x,s):=\exp(s\eta_x)$ is well-defined (for simplicity and without loss of generality, we suppose that $M=\tilde{M}$ in definition \ref{def.1.1}). Moreover, since $L$ is weakly future inextendible along    $\alpha(s):=\widetilde{\Psi}(x_0,s)$,  the pair $(x_0,\alf^+_\eta(x_0))$ is a singular point of $\widetilde{\Psi}$. Therefore, there are curves $x:(-\varepsilon,\varepsilon)\rightarrow \Sigma$ and $t:(-\varepsilon,\varepsilon)\rightarrow (-\delta,\alf^+_\eta(x_0)+\delta)$ such that $x(0) = x_0$, $t(0)=\alf^+_\eta(x_0)$  and
	\begin{align*}
		\widetilde{\Psi}_{*_{(x_0,\alf^+_\eta(x_0))}}(x'(0)+t'(0)\partial_s|_{\alf^+_\eta(x_0) })=0.
	\end{align*}
	In other words, 
	\begin{align*}
		\widetilde{\Psi}_{*_{(x_0,\alf^+_\eta(x_0))}}(x'(0))=\lambda \alpha'(\alf^+_\eta(x_0))\quad\hbox{for some constant $\lambda$.}
	\end{align*}
	Consider the geodesic variation $X:(-\varepsilon,\varepsilon)\times [0,\alf^+_\eta(x_0)]\rightarrow M$, for $\varepsilon>0$ small enough, given by $X(t,s)=\widetilde{\Psi}(x(t),s)=\exp_{x(t)}(s\eta_{x(t)})$. Consider also the Jacobi vector field $J(s):=X_t(0,s)$. Then,
	\begin{align*}
		J(\alf^+_\eta(x_0))=\lambda \alpha'(\alf^+_\eta(x_0)),\qquad J(0)=x'(0),\qquad J'(s)=\nabla_{J(s)}\eta.
	\end{align*}
	Therefore, if we define $f(s):=g(J(s),J(s))$ then 
	$$	f'(s)=2g(\nabla_{J(s)}\eta,J(s))=-2 B_{\eta}(J(s),J(s)),\;\quad f(0)>0,\;\quad f(\alf^+_\eta(x_0))=0.
	$$
	Observe also that there is $\epsilon>0$ such that $f(s)>0$ for all $s\in [\alf^+_\eta(x_0)-\epsilon,\alf^+_\eta(x_0))$. Indeed, take $\epsilon>0$ such that $\alpha([\alf^+_\eta(x_0)-\epsilon,\alf^+_\eta(x_0)])$ is inside a normal neighbourhood of $\alpha(\alf^+_\eta(x_0))$.
	If $f(s_0)=0$ for some $s_0\in[\alf^+_\eta(x_0)-\epsilon,\alf^+_\eta(x_0))$, then $J(s_0)$ would be proportional to $\alpha'(s_0)$. Since $J(\alf^+_\eta(x_0))$ is also proportional to $\alpha'(\alf^+_\eta(x_0))$ and $\alpha(s_0)$ and $\alpha(\alf^+_\eta(x_0))$ are not conjugate points along $\alpha$, then $J(s)$ would be proportional to $\alpha'(s)$ for all $s\in[0,\alf^+_\eta(x_0)]$, contradicting $f(0)>0$.
	
	Since $M$ satisfies the null convergence condition, from inequality  \eqref{ineqRay} we have that $H_\eta(\alpha(s))$ is nondecreasing and thus $H_\eta(\alpha(s))>0$ for all $s\in [0,\alf^+_\eta(x_0))$ because $H_\eta(x_0)>0$. So, we can define 
	\begin{align*}
		A(s):=
		\frac{B_\eta(J(s),J(s))}{g(J(s),J(s))H_\eta(\alpha(s))}\qquad\hbox{for all $s\in (\alf^+_\eta(x_0)-\epsilon,\alf^+_\eta(x_0))$.}
	\end{align*}
	By solving the differential equation $f'(s)=-2f(s)H_\eta(\alpha(s)) A(s)$, 
	we get 
	$$
	f(s)=f(s_0)e^{-2\int_{s_0}^s H_{\eta}(\alpha(u))A(u)du},
	$$
	where $s_0=\alf^+_\eta(x_0)-\epsilon$.
	Since  $A(s)^2\leq \Gamma_\eta^+(x,s)$ and $L$ is of finite type, $|A(s)|$ is necessarily bounded.  
	Taking into account this, that $H_\eta(\alpha(s))$ is nondecreasing and that $f(\alf^+_\eta(x_0))=0$,  we conclude
	$$
	\lim_{s\rightarrow \alf^+_\eta(x_0)}H_{\eta}(\alpha(s))=\infty.
	$$
\end{proof}

\medskip

\begin{proof} {\em (of Theorem \ref{teoprin}).} 
	Let $\zeta\in\mathfrak{X}(M)$ be a future-directed timelike vector field on $M$ and consider the associated future-directed null section $\xi\in \mathfrak{X}(L)$ determined by the condition $g(\zeta,\xi)\equiv -1$. According to Proposition \ref{lemaed} \eqref{p1prop}, it suffices to show that $\Lambda^+_\xi:L\rightarrow \mathbb{R}$ is differentiable. Henceforth, fix $x_0 \in L$.

	By Proposition  \ref{propbanda1} (2), there is an open subset $x_0\in V\subset L$ and a geodesic null section $\eta\in\mathfrak{X}(V)$ with $\eta_{x_0}=\xi_{x_0}$ and such that $\exp_x(s\eta_x)\in V$ for all $x\in V$ and all $0\leq s<\alf^+_\eta(x)$.
	Let $y:\mathcal{B}_\xi^V\rightarrow \mathbb{R}$ be given by $y_x(s):=H_\eta(\exp_x(s\eta_x))$. The Raychaudhuri equation \eqref{eqRay} is written as 
	$$
	R_x(s)=y'_x(s)-y_x(s)^2\Gamma_\eta(x,s),
	$$
	where 
	\begin{align*}
		R_x(s)&:={\rm Ric}\left (  \frac{d}{ds}\exp_x(s\eta_x)   , \frac{d}{ds}\exp_x(s\eta_x)  \right).
	\end{align*}
	Since $L$ has future prolonged null generators, $R_x$ can be extended to a differentiable function on some open subset of $V\times\mathbb{R}$ containing the closure of  $\mathcal{B}_\xi^V$. Since $L$ is of finite type, the same happens for the function $\Gamma_\eta$. 
	Moreover, Proposition \ref{propinfty} ensures that  $\lim_{s\rightarrow \alf^+_\eta(x)}y_x(s)=\infty$, and consequently, $z_x(s):=\frac{1}{y_x(s)}$ satisfies 
	$$
	z'_x(s)+R_x(s)z_x(s)^2=-\Gamma_\eta(x,s),\qquad \lim_{s\rightarrow \alf^+_\eta(x)}z_x(s)=0.
	$$
	If we take $z:\mathcal{O}\rightarrow\mathbb{R}$, where 
	$\mathcal{O}\subset V\times \mathbb{R}$, the maximal solution to the initial value problem
	$$
	\begin{cases}
		z_x'(s)+R_x(s)z_x(s)=-\Gamma_\eta(x,s) \\
		z_x(0)=\frac{1}{H_\eta(x)}
	\end{cases}
	$$
	for $x\in V$, then $\mathcal{B}_\xi^V\subset \mathcal{O}$ and $(\alf_\eta^+(x),x)\in\mathcal{O}-\mathcal{B}_\xi^V$ for all $x\in V$.

	Therefore, $z_x(s)$ can be extended beyond $\alf^+_\eta(x)$ and $z_x(\alf^+_\eta(x))=0$. Moreover\footnote{Here we implicitly use smooth dependence on the parameters. This follows from standard arguments used, for instance, in the proof of \cite[Thm. 9.48. p. 237]{Lee}.}, 
	$z'_{x_0}(\alf^+_\eta(x_0))=-\Gamma_\eta(x_0,\alf^+_\eta(x_0))\neq 0$. Then, by the Implicit Function Theorem, $\alf^+_\eta$ is smooth at a neighbourhood of $x_0$, and by Lemma \ref{cambiofunciones}, the same holds for $\alf^+_\xi$. 
\end{proof}

We have the following as an immediate corollary of Theorem \ref{teoprin}.

\begin{corollary}  Let $(M,g)$ be a spacetime which obeys the null convergence condition and let $L\subset M$ be a null hypersurface such that:
	\begin{enumerate}
		\item $L$ is future  weakly inextendible and it has future  prolonged null generators.
		\item $L$ is totally umbilic with positive null mean curvature (with respect to a future-directed null section).
	\end{enumerate}
	Then, there exists a  null section with nonzero constant surface gravity.
\end{corollary}

By combining (the proof of) Theorem \ref{teoprin}
with Proposition \ref{lemaed} \eqref{pgprop}, we also deduce the following result.

\begin{corollary}  Let $(M,g)$ be a spacetime which obeys the null convergence condition and $L\subset M$ a  null hypersurface such that:
	\begin{enumerate}
		\item $L$ is future and past weakly inextendible and it has future and past prolonged null generators.
		\item $L$ is of future and past finite type. 
	\end{enumerate}
	Then, there exists a geodesic null section.
\end{corollary}

Observe that the second condition in the above corollary implies 
that along each null generator there are points with positive and negative null mean curvature.

\section*{Acknowledgements}
The authors are partially supported by the project PID2020-118452GBI00. The second author is also partially supported by the IMAG-María de Maeztu grant CEX2020-001105-M (funded by MCIN/AEI/10.13039/50110001103).

\end{document}